\documentclass[12pt,reqno]{amsart}

\usepackage[utf8]{inputenc}
\usepackage{amsmath, amssymb, amsthm}
\usepackage{amsfonts}
\usepackage{color}
\usepackage{bm,bbm}
\usepackage{ascmac,paralist}

\usepackage{mathtools}
\mathtoolsset{showonlyrefs}

\usepackage{hyperref}
\hypersetup{
setpagesize=false,
 colorlinks=true,%
 linkcolor=magenta,
 citecolor=blue,
}
\newtheorem{theorem}{Theorem}[section]

\newtheorem{lemma}[theorem]{Lemma}

\usepackage{subfiles}
\DeclareMathOperator{\sgn}{sgn}

\newcommand{\ol}[1]{\overline{#1}}
\newcommand{\lr}[1]{\left( #1 \right)}
\newcommand{\Lr}[1]{\left\{ #1 \right\}}

\newcommand{\Z}{\mathbb{Z}}
\newcommand{\R}{\mathbb{R}}
\newcommand{\C}{\mathbb{C}}
\newcommand{\T}{\mathbb{T}}

\newcommand{\cK}{\mathcal{K}}
\newcommand{\cH}{\mathcal{H}}

\renewcommand{\Re}{\mathrm{Re}\,}
\renewcommand{\Im}{\mathrm{Im}\,}
\newcommand{\ind}{\mathrm{ind}\,}

\newcommand{\ess}{\sigma_{\mathrm{ess}}}

\newcommand{\sign}{\mathrm{sgn}\,}

\newcommand{\Usuz}{U_{\textnormal{suz}}}
\newcommand{\suz}{\textnormal{suz}}

\newcommand{\textbi}[1]{\textit{{#1}}}

\usepackage{verbatim}
\usepackage{float}
\usepackage{graphicx}
\usepackage{pgfplots}
	\pgfplotsset{compat=1.12}
\usepackage{tikz}
 	\usetikzlibrary{decorations.markings}
 	\usetikzlibrary{calc}

\newcommand{\sigmap}{\sigma_{\mathrm{p}}}
\usepackage[font={footnotesize,it}]{caption}

\title{Spectral Analysis of Non-unitary Two-phase Quantum Walks in one dimension}

\author{Chusei Kiumi}
\address{Graduate School of Engineering Science, Yokohama National University, Hodogaya, Yokohama, 240-8501, Japan}
\email{kiumi-chusei-bf@ynu.jp}

\author{Kei Saito}
\address{Department of Information Systems Creation, Faculty of Engineering, Kanagawa University, Kanagawa, Yokohama, 221-8686, Japan}
\email{ft102130ev@jindai.jp}

\author{Yohei Tanaka}
\address{Department of Mathematics, Shinshu university, Matsumoto 390-8621, Japan}
\email{tana35@shinshu-u.ac.jp}

\begin{document}

\begin{abstract}
It is recently shown by Asahara-Funakawa-Seki-Tanaka that existing index theory for chirally symmetric (discrete-time) quantum walks can be extended to the setting of non-unitary quantum walks. More precisely, they consider a certain non-unitary variant of the two-phase split-step quantum walk as a concrete one-dimensional example, and give a complete classification of the associated index in their study. Note, however, that it remains uncertain whether or not their index gives a lower bound for the number of so-called topologically protected bound states unlike the setting of unitary quantum walks. In fact, the spectrum of a non-unitary operator can be any subset of the complex plane, and so the definition of such bound states is ambiguous in the non-unitary case. The purpose of the present article is to show that the simple use of transfer matrices naturally allows us to obtain an explicit formula for a topologically bound state associated with the non-unitary split-step quantum walk model mentioned above.
\end{abstract}

\maketitle


\section{Introduction}
\label{section: introduction}

The major theme of the present article belongs to the broad class of spectral analysis for non-unitary quantum walks. Quantum walk theory is a quantum-mechanical counterpart of the classical random walk theory \cite{Gudder-1988,Aharonov-Davidovich-Zagury-1993,Meyer-1996,Ambainis-Bach-Nayak-Vishwanath-Watrous-2001}, and many useful applications of this ubiquitous concept can be found in \cite{Portugal-2013}. From a purely mathematical point of view, we may regard (discrete-time) quantum walks as periodically driven systems. Here, the one-step evolution after one driving period, known as a \textbi{time-evolution operator}, is given by a fixed unitary operator $U$ on a Hilbert space $\cH.$ We may consider various symmetry types for $U.$ For example, the operator $U$ is said to have \textbi{chiral symmetry}, if the following equality holds true;
\begin{equation}
\label{equation: chiral symmetry}
U^* = \varGamma U \varGamma,
\end{equation}
where $\varGamma$ is a fixed unitary self-adjoint operator on $\cH.$ In the presence of symmetries of this kind, index theory for \textit{unitary} quantum walks on the one-dimensional integer lattice $\Z$ has been a particularly active theme of recent mathematical studies of quantum walks \cite{Cedzich-Grunbaum-Stahl-Velazquez-Werner-Werner-2016,Cedzich-Geib-Grunbaum-Stahl-Velazquez-Werner-Werner-2018,Cedzich-Geib-Stahl-Velazquez-Werner-Werner-2018}. In particular, the chiral symmetry condition \eqref{equation: chiral symmetry} alone has attracted tremendous attention \cite{Suzuki-2019,Suzuki-Tanaka-2019,Matsuzawa-2020,Tanaka-2020,Cedzich-Geib-Werner-Werner-2021}, and it is also the main subject of the present article. 

\textbi{Suzuki's split-step quantum walk} \cite{Fuda-Funakawa-Suzuki-2017,Fuda-Funakawa-Suzuki-2018,Fuda-Funakawa-Suzuki-2019,Tanaka-2020,Narimatsu-Ohno-Wada-2021, Fuda-Narimatsu-Saito-Suzuki-2022} can be viewed as a standard example of a one-dimensional unitary quantum walk satisfying \eqref{equation: chiral symmetry}. The time-evolution operator of this model is given by the following $2 \times 2$ block-operator matrix on 
$\mathcal{H}=\ell^2(\mathbb{Z};\mathbb{C}^2)\simeq \ell^2(\mathbb{Z};\mathbb{C})\oplus\ell^2(\mathbb{Z};\mathbb{C}),$
where $\ell^2(\Z;\mathbb{C})$ is the Hilbert space of square-summable $\C$-valued sequences indexed by $\Z:$

\begin{align}
\label{equation: definition of Suzuki evolution operator}
U_{\textnormal{suz}} 
:= 
\begin{pmatrix}
1 & 0 \\
0  &  L^*
\end{pmatrix}
\begin{pmatrix}
p & \sqrt{1 - p^2} \\
\sqrt{1 - p^2} & -p
\end{pmatrix}
\begin{pmatrix}
1 & 0 \\
0  &  L
\end{pmatrix}
\begin{pmatrix}
a & \sqrt{1 - a^2} \\
\sqrt{1 - a^2} & -a
\end{pmatrix},
\end{align}
where $L$ is the bilateral left-shift operator on $\ell^2(\Z) := \ell^2(\Z;\mathbb{C})$ defined by $L \Psi = \Psi(\cdot + 1)$ for each $\Psi \in \ell^2(\Z),$ and where $p = (p(x))_{x \in \Z}$ and $a = (a(x))_{x \in \Z}$ are real-valued sequences assuming values in the closed interval $[-1,1].$ Note that any bounded sequence indexed by $\Z$ is identified with the corresponding multiplication operator on $\ell^2(\Z)$ throughout the present article. The unitary operator $\Usuz$ can then be decomposed into the product $\Usuz = S_\suz C_\suz,$ where the \textbi{shift operator} $S_\suz$ and \textbi{coin operator} $C_\suz$ are given respectively by
\[
S_\suz
:= 
\begin{pmatrix}
1 & 0 \\
0  &  L^*
\end{pmatrix}
\begin{pmatrix}
p & \sqrt{1 - p^2} \\
\sqrt{1 - p^2} & -p
\end{pmatrix}
\begin{pmatrix}
1 & 0 \\
0  &  L
\end{pmatrix}, \quad 
C_\suz 
:= 
\begin{pmatrix}
a & \sqrt{1 - a^2} \\
\sqrt{1 - a^2} & -a
\end{pmatrix}.
\]
If we set $(\varGamma, U) := (S_\suz ,U_{\textnormal{suz}})$ or $(\varGamma, U) := (C_\suz ,U_{\textnormal{suz}}),$ then it is easy to verify that the chiral symmetry condition \eqref{equation: chiral symmetry} holds true, since $S_\suz, C_\suz$ are unitary self-adjoint. It is well-known that in either case we can assign to the given pair $(\varGamma, U)$ a certain well-defined Fredholm index, say $\ind(\varGamma, U)$ (see \textsection \ref{section: bulk-edge} for precise definition). Moreover, the following estimate holds true;
\begin{equation}
\label{equation1: topological protection of bounded states}
|\ind(\varGamma,U)| \leq \dim \ker(U - 1) + \dim \ker(U + 1),
\end{equation}
provided that the essential spectrum of $U,$ denoted by $\ess(U),$ contains neither $-1$ nor $+1.$ As can be seen from \eqref{equation1: topological protection of bounded states}, if $\ind(\varGamma,U)$ is non-zero, then at least one of $\ker(U - 1),\ker(U + 1)$ contains a non-trivial eigenstate known as a \textit{topologically protected bound sate}. To put \eqref{equation1: topological protection of bounded states} into context, we may assume the existence of the following limit for each $\xi = p,a$ and each $\star = -\infty, +\infty:$
\begin{align}
\label{equation: anisotropic assumption}
\xi(\star) := \lim_{x \to \star} \xi(x).
\end{align}  
It is shown in \cite[Theorem B(i)]{Tanaka-2020} that under the assumption \eqref{equation: anisotropic assumption}, we have $-1,+1 \notin \ess(U_\suz)$ if and only if $|p(\star)| \neq |a(\star)|$ for each $\star = -\infty, +\infty.$ Moreover, if we set $(\varGamma, U) := (S_\suz ,U_{\textnormal{suz}}),$ then
{\footnotesize
\begin{equation}
\label{equation: witten index for ssqw}
\ind (\varGamma, U) =
\begin{cases}
0, &  |p(-\infty)| < |a(-\infty)| \mbox{ and }  |p(+\infty)| < |a(+\infty)|, \\
+\sign p(+\infty), &  |p(-\infty)| < |a(-\infty)| \mbox{ and }  |p(+\infty)| > |a(+\infty)|, \\
-  \sign p(-\infty), &  |p(-\infty)| > |a(-\infty)| \mbox{ and } |p(+\infty)| < |a(+\infty)|, \\
+\sign p(+\infty) - \sign p(-\infty), &  |p(-\infty)| >  |a(-\infty)| \mbox{ and }  |p(+\infty)| > |a(+\infty)|,
\end{cases}
\end{equation}
}
where $\sign : \R \to \{-1, 0, 1\}$ denotes the sign function.  Note that \eqref{equation: witten index for ssqw} is robust in the sense that it depends only on the asymptotic values \eqref{equation: anisotropic assumption}. An analogous index formula for the case $(\varGamma, U) := (C_\suz ,U_{\textnormal{suz}})$ can be found in \cite[Theorem B(i)]{Tanaka-2020}.

The index formula \eqref{equation: witten index for ssqw} is revisited in \cite{Asahara-Funakawa-Seki-Tanaka-2020} with $U_\suz$ replaced by a certain non-unitary version, say, $\tilde{U}_\suz$ (see \textsection \ref{section: preliminaries} for details). The shift operator of $\tilde{U}_\suz$ remains unchanged, but the new coin operator is non-unitary as it contains an additional $\R$-valued sequence $\gamma = (\gamma(x))_{x \in \Z}.$ From a physical point of view, the new parameter $\gamma$ represents the gain-loss effects of photons in an optical network experiment setup for the non-unitary model introduced by Mochizuki-Kim-Obuse \cite{Mochizuki-Kim-Obuse-2016} (see \cite[Theorem B]{Asahara-Funakawa-Seki-Tanaka-2020} for details). If $\gamma$ is identically zero, then such effects are deemed to be negligible, and we recover the unitary case $\tilde{U}_\suz = U_\suz.$ Otherwise, $\tilde{U}_\suz$ is non-unitary in general. Under the assumption that limits of the form \eqref{equation: anisotropic assumption} exist for each $\xi = p,a, \gamma$ and each $\star = -\infty, +\infty,$ an analogous index formula for $\tilde{U}_\suz$ can be proved (see \cite[Theorem C(i)]{Asahara-Funakawa-Seki-Tanaka-2020} for details).


Note, however, that it is not entirely obvious if this index formula plays any physically important role. In fact, it remains unclear whether or not \eqref{equation1: topological protection of bounded states} still holds true for \textit{non-unitary} $U.$ Note that the spectrum of a non-unitary operator can be any subset of the complex plane, and so the definition of topologically protected bound states is ambiguous in general. The purpose of the present article is to show that the use of \textit{transfer matrices} naturally allows us to obtain an explicit formula for such a bound state associated with $\tilde{U}_\suz.$ The utility of transfer matrices is known to be particularly useful in characterising eigenvalues and eigenstates, but they can also be used to clarify various properties of quantum walks, such as stationary measures, dispersive estimate, or spectral types of time-evolution
\cite{Cedzich-Fillman-Ong-2021, Maeda-Sasaki-Segawa-Suzuki-Suzuki-2022, Komatsu-Konno-2022, Kawai-Komatsu-Konno-2022}.

The rest of the present article is organised as follows. In \textsection \ref{section: preliminaries} we give the precise definition of the non-unitary split-step quantum walk $\tilde{U}_\suz,$ and transform the associated eigenvalue equation into a first-order difference equation by making use of transfer matrices. \textsection \ref{section: main result} is devoted to proving the main theorem of the present article (Theorem \ref{main_theorem}). The paper concludes with some discussions and remarks in \textsection \ref{section: discussion}.


On a final note, the scope of the present article is beyond that of \cite{Asahara-Funakawa-Seki-Tanaka-2020} which is in essence a study of topological invariants on the integer lattice $\Z,$ namely, mathematical quantities which depend only on asymptotic values of the form \eqref{equation: anisotropic assumption}. On the other hand, eigenstates of $\tilde{U}_\suz$ depend also on local information of $\xi = p, a, \gamma,$ and this is precisely why we need to take a completely different approach in this paper.


\section{Preliminaries}
\label{section: preliminaries}
By an operator we shall always mean an everywhere-defined linear bounded operator on a Hilbert space. The underlying Hilbert space of the present article is $\mathcal{H}=\ell^2(\mathbb{Z};\mathbb{C}^2)\simeq \ell^2(\mathbb{Z})\oplus\ell^2(\mathbb{Z}).$ The non-unitary version of Suzuki' split-step quantum walk \eqref{equation: definition of Suzuki evolution operator} we shall consider in this paper is an operator $U$ on $\cH$ of the form $U = SC,$ where the \textbi{shift operator} $S$ and \textbi{coin operator} $C$ are given respectively by the following formulas:
\begin{align*}
    &S := \begin{pmatrix}
    p(\cdot) & q(\cdot)L \\ L^*q^*(\cdot) & -p(\cdot -1)
    \end{pmatrix}
    , &
    &C := \begin{pmatrix}
    e^{-2\gamma(\cdot + 1)}a(\cdot) & e^{\gamma(\cdot) - \gamma(\cdot + 1)}b^*(\cdot)
    \\
    e^{\gamma(\cdot) - \gamma(\cdot + 1)}b(\cdot) & -e^{2\gamma}a(\cdot)
    \end{pmatrix},
\end{align*}
where we assume that three real-valued sequences $\gamma = (\gamma(x))_{x \in \Z}, p = (p(x))_{x \in \Z}, a = (a(x))_{x \in \Z},$ and two complex-valued sequences $q = (q(x))_{x \in \Z}, b = (b(x))_{x \in \Z}$ satisfy the following two conditions for each $x \in \Z:$
\begin{align*}
&p(x), a(x) \in (-1,1), \\
& p^2(x)+|q(x)|^2=a^2(x)+|b(x)|^2 = 1.
\end{align*}
We also assume that the coin operator $C$ has a \textit{two-phase} in the sense that for each $x \in \Z$ the two parameters $a(x), b(x)$ are of the following forms:
\begin{align*}
    &a(x) = \begin{cases}
    a_p, \quad & x\geq 0,
    \\
    a_m, \quad & x<0,
    \end{cases}
    &
    &b(x) = \begin{cases}
    b_p, \quad & x\geq 0,
    \\
    b_m, \quad & x<0,
    \end{cases}
\end{align*}
where $a_p, a_m \in (-1,1).$ For simplicity, we assume that all parameters of the shift operator $S$ are uniform. That is, for each $\xi = \gamma, p, q$ and each $x \in \Z,$ we set $\xi(x) = \xi.$

We consider en eigenvalue equation of the form $U \Psi = \lambda \Psi$ in what follows, where we may assume $\lambda \neq 0$ without loss of generality. Indeed, the set of eigenvalues of $U,$ denoted by $\sigmap(U),$ does not contain $0,$ since $U = SC$ is invertible. 



Any solution $\Psi \in \ker(U-\lambda)$ admits the following representation as in \cite{Kiumi-Saito-2021}:
\begin{align}
    \Psi(x) = 
    \begin{cases}
    \lr{T_p}^x T_{-1}\varphi,\quad & x\geq 0,
    \\
    \varphi,\quad & x=-1,
    \\
    \lr{T_m}^{x+1}\varphi,\quad & x\leq -2,
    \end{cases}
    \label{0201A}
\end{align}
for some $\varphi\in\mathbb{C}^2.$ Of course, $\varphi$ becomes the zero vector, if $\lambda \notin \sigmap(U).$ Here, the $2 \times 2$ matrices $T_p, T_m, T_{-1}$ defined below are called \textit{transfer matrices}: 
\begin{align}
T_{\ast} &=
\frac{1}{b_{\ast}q \lambda}
\begin{pmatrix}
\lambda^2 -2a_{\ast} p\cosh( 2\gamma)\lambda  +a_{\ast}^{2}
& -\ol{b_{\ast}}\lr{ p\lambda - a_{\ast} e^{2\gamma}}
\\
-b_{\ast}\lr{ p\lambda - a_{\ast} e^{-2\gamma}}
& |b_{\ast}|^{2}
\end{pmatrix},
\quad \ast\in \{p, m\}
\\[+8pt]
T_{-1} &=
\frac{1}{b_p q \lambda}
\begin{pmatrix}
\lambda^2 - \lr{a_p e^{2\gamma} + a_m e^{-2\gamma}}p\lambda + a_m a_p
&
-\ol{b_m}\lr{p - a_p e^{2\gamma}}
\\
-b_p\lr{p - a_m e^{-2\gamma}}
&
\ol{b_m}b_p
\end{pmatrix}.
\end{align}
Note that the existence of a non-trivial vector $\varphi\in\mathbb{C}^2$ in \eqref{0201A} is a necessary and sufficient condition for $\lambda\in\sigmap(U).$ To ensure the existence of such a vector $\varphi,$ we will impose some restrictions on the eigenvalue $\lambda.$

Firstly, we can easily check $\det|T_\ast| = 1$ for each $\ast\in\{p,m\},$ and so the two eigenvalues of the transfer matrix $T_\ast,$ denoted by $\zeta_\ast^>$ and $\zeta_\ast^<,$ must satisfy $|\zeta_\ast^>| \neq |\zeta_\ast^<|.$ If not, the sequence $\Psi$ in \eqref{0201A} fails to be square-summable unless $\varphi=0.$  By a direct calculation, the eigenvalues of each $T_\ast$ 
are given by
\begin{align*}
    \zeta_\ast^> &= \frac{\lr{\lambda + \lambda^{-1}-2p a_\ast \cosh(2\gamma)} + s_\ast
    \sqrt{\lr{\lambda + \lambda^{-1}-2p a_\ast \cosh(2\gamma)}^2 - 4|b_\ast q|^2}
    }{2b_\ast q}
    ,
    \\
    \zeta_\ast^< &= \frac{\lr{\lambda + \lambda^{-1}-2p a_\ast \cosh(2\gamma)} - s_\ast
    \sqrt{
    \lr{\lambda + \lambda^{-1}-2p a_\ast \cosh(2\gamma)}^2 - 4|b_\ast q|^2}
    }{2b_\ast q}
    ,
\end{align*}
where $s_\ast$ is a sign function such that $|\zeta_\ast^>| \geq |\zeta_\ast^<|$ holds. We note that $\lambda$ is always real according to the later discussion in this paper (see Lemma \ref{lem:lambda} for details), and so this sign function is given explicitly by $s_\ast = \sgn\lr{\Re\lr{\lambda + \lambda^{-1}}-2p a_\ast \cosh(2\gamma)},$ where $\Re(x)$ denotes the real part of a complex number $x.$ It follows that the following set contains the point spectrum $\sigmap(U);$
\[
    \Lambda := 
    \Lr{\lambda\in\mathbb{C} \mid |\zeta_\ast^>| \neq |\zeta_\ast^<|,\ \ast\in\{p, m\}}.
\]
\begin{lemma}
\label{lem:Lambdacheck}
For each $\lambda \in \mathbb{C}\setminus\{0\}$ and each $\ast\in\{p, m\},$ we have $\lambda \not\in \Lambda $ if and only if
\[
\lambda + \lambda^{-1}\in\mathbb{R}\quad \text{and}\quad 
\lr{\lambda +\lambda ^{-1} -2p a_\ast\cosh( 2\gamma)}^{2} -4| b_{\ast}q| ^{2} \leq 0.
\]
\end{lemma}
\begin{proof}
Let $X=\lambda +\lambda^{-1} -2p a_\ast\cosh( 2\gamma).$
If $|\zeta_\ast^>| = |\zeta_\ast^<|,$ then these equal to $1$ since $|\det(T_\ast)|=1.$
Thus, holding this equation if and only if
\[
    \left|X\pm \sqrt{X^{2} -4|b_\ast q| ^{2}}\right| ^{2} =4|b_\ast q| ^{2}
\]
for both signs $+$ and $-.$
Moreover, these conditions are equivalent to holding following three conditions (i) to (iii) by taking square for both side after a deformation.
\begin{enumerate}
    \item[(i)] $\left|X^2 - 4|b_\ast q|^2\right|^2 = \left||X|^2 - 4|b_\ast q|^2\right|^2,$
    \item[(ii)] $\Re\lr{\ol{X}\sqrt{X^2 - 4 |b_\ast q|^2}}=0,$
    \item[(iii)] $|X|^2 -4|b_\ast q|^2 \leq 0.$
\end{enumerate}
Here, the above condition (i) is equivalent to $X\in \mathbb{R}.$
If (i) and (iii) hold, then (ii) holds because $X^2 = |X|^2.$
Thus, we get desired conclusion.
\end{proof}

Noting that $\varphi \in \ker\lr{T_m - \zeta_m^>} \cap \ker\lr{\lr{T_p -\zeta_p^<}T_{-1}}$ is a necessary condition for $\Psi$ in \eqref{0201A} to be square-summable, we conclude the current section with the following lemma;
\begin{lemma}
\label{lem:condition_ev}
We have $\lambda\in\sigmap(U)$ if and only if the following two conditions hold simultaneously:
\begin{enumerate}
    \item[{\rm (i)}] $\lambda \in \Lambda,$
    \item[{\rm (ii)}] $\dim\lr{\ker\lr{T_m - \zeta_m^>} \cap \ker\lr{\lr{T_p -\zeta_p^<}T_{-1}}}\neq 0.$
\end{enumerate}
\end{lemma}

\section{Spectral analysis of the non-unitary split-step quantum walk}
\label{section: main result}

In this section, we focus on a non-trivial characterisation of the condition (ii) in Lemma \ref{lem:condition_ev} to find an explicit formula for $\lambda.$ To do so, we set the following notation for each $ * \in\{m,p\}:$
\begin{align*}
 K_{*} = p\lambda -a_{*}e^{-2\gamma } 
 ,\qquad
 K_{*}^{\prime } = a_{*} e^{2\gamma } -p\lambda ^{-1},
 \qquad
  J_{*} =\lambda -\lambda ^{-1} +2pa_{*}\sinh( 2\gamma ).
\end{align*}
Here, we note that
\begin{align}
\label{eq:J^2-4KK'}
\left( \lambda +\lambda ^{-1} -2p a_{*}\cosh( 2\gamma)\right)^{2} -4| b_{*}q| ^{2} =J_{*}^{2} -4K_{*} K_{*}^{\prime }.
\end{align}

For $\lambda\in\mathbb{C}\setminus\{0\},$ let the two eigenvectors of $T_\ast$ associated with $\zeta _{*}^{>},\zeta _{*}^{<}$ be $v_{*}^{ >},v_{*}^{<}$ respectively. We may assume without loss of generality that $J_\ast \neq 0.$ Otherwise, $\zeta_\ast^> = \zeta_\ast^<$ implies $\lambda\not\in \Lambda,$ and so $\lambda \notin \sigmap(U)$ according to Lemma \ref{lem:condition_ev}. In the case of $K_{*} \neq 0,$ the two vectors $v_{*}^{ >}$ and $v_{*}^{<}$ can be written as
\begin{align}
\label{eq:v>}
v_{*}^{ >} =\begin{pmatrix}
-\dfrac{\lambda }{2b_{*}}\left( J_{*} +s_{*}\sqrt{J_{*}^{2} -4K_{*} K_{*}^{\prime }} -2a_{*} e^{2\gamma } \lambda ^{-1} K_{*}\right)\\
K_{*}
\end{pmatrix}
\end{align}
and
\begin{align}
\label{eq:v<}
v_{*}^{< } =\begin{pmatrix}
-\dfrac{\lambda }{2b_{*}}\left( J_{*} -s_{*}\sqrt{J_{*}^{2} -4K_{*} K_{*}^{\prime }} -2a_{*} e^{2\gamma } \lambda ^{-1} K_{*}\right)\\
K_{*}
\end{pmatrix}.
\end{align}
In the case of $ K_{*} =0,$ we have $J_\ast\in\mathbb{R}$ and 
\[
T_{*} =\frac{1}{b_{*} q\lambda }\begin{pmatrix}
\lambda J_{*} +|b_{*} |^{2} & 2a_{*}\overline{b_{*}} \sinh( 2\gamma )\\
0 & | b_{*}| ^{2}
\end{pmatrix}.
\]
Moreover, the eigenvalues of the transfer matrix $T_{*}$ become
\[
\zeta _{*}^{ >} =\frac{( 1+s_{*}\sgn( J_{*})) \lambda J_{*} +2|b_{*} |^{2}}{2b_{*} q\lambda } ,\quad \zeta _{*}^{< } =\frac{( 1-s_{*}\sgn( J_{*})) \lambda J_{*} +2|b_{*} |^{2}}{2b_{*} q\lambda }.
\]
Note that the $(2,2)$-element of the matrix $T_\ast - \zeta_\ast^>$ (resp. $T_\ast - \zeta_\ast^<$) equals to zero if and only if $s_\ast \sgn(J_\ast)<0$ (resp. $s_\ast \sgn(J_\ast)>0$). Thus, $v_{*}^{>}$ and $v_{*}^{<}$ can be written as below:
\begin{align}
\label{eq:eigenvector}
v_{*}^{ >} =\begin{cases}
\begin{pmatrix}
1\\
0
\end{pmatrix} , & s_{*}\sgn( J_{*})  >0,\\[20pt]
\begin{pmatrix}
2a_{*}\overline{b_{*}}\sinh( 2\gamma )\\
-\lambda J_{*}
\end{pmatrix} , & s_{*}\sgn( J_{*}) < 0,
\end{cases}\ v_{*}^{< } =\begin{cases}
\begin{pmatrix}
2a_{*}\overline{b_{*}}\sinh( 2\gamma )\\
-\lambda J_{*}
\end{pmatrix} , & s_{*}\sgn( J_{*})  >0,\\[20pt]
\begin{pmatrix}
1\\
0
\end{pmatrix} , & s_{*}\sgn( J_{*}) < 0.
\end{cases}
\end{align}
\begin{lemma}
\label{lem:cond}
For $\lambda\in\mathbb{C}\setminus\{0\},$ the condition (ii) in Lemma \ref{lem:condition_ev}, i.e., \\$ \dim\lr{\ker\lr{T_m - \zeta_m^>} \cap \ker\lr{\lr{T_p -\zeta_p^<}T_{-1}}}\neq 0,$ is satisfied if and only if the following three conditions hold simultaneously:
\begin{enumerate}
    \item [{\rm (i)}]$( K_{m} J_{p} -K_{p} J_{m})\left( K_{m}^{\prime } J_{p} -K_{p}^{\prime } J_{m}\right) +\left( K_{m} K_{p}^{\prime } -K_{m}^{\prime } K_{p}\right)^{2} =0,$
    \item [{\rm (ii)}]$\sgn\left( s_{p} K_{m}\sqrt{J_{p}^{2} -4K_{p} K_{p}^{\prime }} +s_{m} K_{p}\sqrt{J_{m}^{2} -4K_{m} K_{m}^{\prime }}\right) =\sgn( K_{m} J_{p} -K_{p} J_{m}),$
    \item [{\rm (iii)}]$\sgn\left( 2K_{m} K_{p}^{\prime } +2K_{p} K_{m}^{\prime } -J_{p} J_{m}\right) =\sgn\left( s_{p} s_{m}\sqrt{J_{p}^{2} -4K_{p} K_{p}^{\prime }}\sqrt{J_{m}^{2} -4K_{m} K_{m}^{\prime }} \ \right).$
\end{enumerate}


\end{lemma}
\begin{proof}
The statement will be proved by checking the condition of linear dependency between $T_{-1}v_m^>$ and $v_p^<$ in \eqref{eq:v>}, \eqref{eq:v<} and \eqref{eq:eigenvector}.
In the case $ K_{m} \neq 0\ $and$ \ K_{p} \neq 0,$
\[
T_{-1} v_{m}^{ >} =\zeta _{m}^{ >}\begin{pmatrix}
-\dfrac{\lambda }{2b_{p}}\left( J_{m} +s_{m}\sqrt{J_{m}^{2} -4K_{m} K_{m}^{\prime }} -2a_{p} e^{2\gamma } \lambda ^{-1} K_{m}\right)\\[+12pt]
K_{m}
\end{pmatrix}
\]
holds. Therefore, $ T_{-1} v_{m}^{ >}$ and $ v_{p}^{< }$ are linear dependent if and only if 
\[
s_{p} K_{m}\sqrt{J_{p}^{2} -4K_{p} K_{p}^{\prime }} +s_{m} K_{p}\sqrt{J_{m}^{2} -4K_{m} K_{m}^{\prime }} =K_{m} J_{p} -K_{p} J_{m}.
\]
By squaring both sides, we can check that this equation is equivalent to 
\begin{align}
    \label{eq:proof1}
    2K_{m} K_{p}^{\prime } +2K_{p} K_{m}^{\prime } -J_{p} J_{m} =s_{p} s_{m}\sqrt{J_{p}^{2} -4K_{p} K_{p}^{\prime }}\sqrt{J_{m}^{2} -4K_{m} K_{m}^{\prime }}
\end{align}
and (ii). 
Furthermore, by squaring both sides and removing the square root again, \eqref{eq:proof1} is converted to the equations (i) and (iii).

Next, we consider the case $ K_{m} =0$ and $ K_{p} \neq 0.$ According to \eqref{eq:eigenvector}, we can see that
\[
T_{-1} v_{m}^{ >} =\begin{cases}
\dfrac{\overline{q} \lambda }{b_{p}}\begin{pmatrix}
1\\
0
\end{pmatrix} , & s_{m}\sgn( J_{m})  >0,\\[20pt]
\dfrac{\overline{b_{m}}}{b_{p} q}\begin{pmatrix}
\lambda K_{m}^{\prime } -a_{p} e^{2\gamma } J_{m}\\
-b_{p} J_{m}
\end{pmatrix} , & s_{m}\sgn( J_{m}) < 0.
\end{cases}
\]
 If $ s_m \sgn(J_{m}) > 0,$ it can be easily checked that $T_{-1}v_{m}^{ >}$ and $ v_{p}^{< }$ are always not linear dependent since $ K_{p} \neq 0.$ Thus, in $ K_{m} =0$ and $ K_{p} \neq 0$ case, $ T_{-1} v_{m}^{ >}$ and $ v_{p}^{< }$ are linear dependent if and only if
\[
 \left( J_{p} -s_{p}\sqrt{J_{p}^{2} -4K_{p} K_{p}^{\prime }}\right) J_{m} =2K_{p} K_{m}^{\prime }
\]
and
\[
 s_{m}\sgn( J_{m}) < 0,
\]
which are equivalent to \eqref{eq:proof1} and (ii), respectively.
As mentioned above, \eqref{eq:proof1} is equivalent to (i) and (iii), so the proof of this case is completed.
Considering similarly, we can also prove the statement with respect to the case $ K_{m} \neq 0,\ K_{p} =0.$

Finally, we consider the case $K_{m} =K_{p} =0.$
Since $a_{m} =a_{p}$ holds, we can write $s=s_{m} =s_{p}$ and $J=J_{m} =J_{p} .$ From \eqref{eq:eigenvector}, $T_{-1} v_{m}^{ >}$ and $v_{p}^{< }$ become
\begin{align*}
    \lr{T_{-1}v_m^>,\ v_p^<}
    =
    \begin{cases}
    \lr{
    \dfrac{\overline{q} \lambda }{b_{p}}
    \begin{pmatrix}
    1\\
    0
    \end{pmatrix}
    ,\quad
    \begin{pmatrix}
    2a\overline{b_{p}}\sinh( 2\gamma )\\
    -\lambda J
    \end{pmatrix}
    }
    ,\quad &s \sgn(J)>0,
    \\[+20pt]
    \lr{
    \dfrac{\overline{b_{m}}}{b_{p} q}\begin{pmatrix}
    \lambda K_{m}^{\prime } -ae^{2\gamma } J\\
    -b_{p} J
    \end{pmatrix}
    ,
    \quad
    \begin{pmatrix}
1\\
0
\end{pmatrix}
}
,
\quad &s \sgn(J)<0.
    \end{cases}
\end{align*}
Since $\lambda, J \neq 0,$ these vectors are not linear dependent.
Moreover, the equation (iii) does not hold in this case because this equation becomes $-1 = s_p s_m,$ but $s_p = s_m$ gives $s_p s_m=1.$
From the above discussions, the proof is complete.
\end{proof}
By analysing the equations (i), (ii) and (iii) in Lemma \ref{lem:cond}, we can explicitly calculate the eigenvalues of $U.$
\begin{lemma}
\label{lem:lambda}
Let $p\cosh(2\gamma)/|q| \neq s_2 a_\ast/|b_\ast|$ for each $\ast\in\{m, p\},$ and let 
\[
\lambda(s_1, s_2):=s_1p \sinh(2\gamma)+s_2\sqrt{1+p^2\sinh^2(2\gamma)}, \quad (s_1, s_2) \in \{-1, +1\}^2.
\]
Then we have the following inclusions:
\[
\sigmap(U)\subset\{\lambda(s_1, s_2)\mid (s_1,s_2)\in\{-1, +1\}^2\}\subset\Lambda.
\]
Moreover, $\lambda=\lambda(s_1, s_2)$ is equivalent to the condition (i) in Lemma \ref{lem:cond}.
\end{lemma}
\begin{proof}
Firstly, we will show $\{\lambda(s_1, s_2)\mid (s_1,s_2)\in\{-1, +1\}^2\}\subset\Lambda.$
By the Lemma \ref{lem:Lambdacheck}, for $\lambda = \lambda(s_1, s_2),$ it is sufficient to prove $\lr{\lambda + \lambda^{-1}-2pa_\ast \cosh(2\gamma)}^2 - 4|b_\ast q|^2> 0.$
Here, \eqref{eq:J^2-4KK'} and direct calculation give
\begin{align*}
\lr{\lambda + \lambda^{-1}-2pa_\ast \cosh(2\gamma)}^2 - 4|b_\ast q|^2 &= J_\ast^2-4K_\ast K'_\ast
\\
&= 4\lr{ p\cosh( 2\gamma ) -s_{2} a_\ast \sqrt{|q|^2+p^{2}\cosh^{2}( 2\gamma )}}^{2}.
\end{align*}
Thus, the above inequality is always true except $p\cosh(2\gamma)/|q| \neq s_2 a_\ast/|b_\ast|$ case, so $\{\lambda(s_1, s_2)\mid (s_1,s_2)\in\{-1, +1\}^2\}\subset\Lambda$ is proved.

Next, we will show $\sigmap(U)\subset\{\lambda(s_1, s_2)\mid (s_1,s_2)\in\{-1, +1\}^2\}.$
Lemma \ref{lem:condition_ev} and Lemma \ref{lem:cond} say the condition (i) in Lemma \ref{lem:cond} is a necessary condition for $\lambda\in\sigmap(U).$
By direct calculation, we have
\begin{align*}
 K_{m} K_{p}^{\prime } -K_{m}^{\prime } K_{p} &=p( a_{m} -a_{p})\left( e^{-2\gamma } \lambda ^{-1} -e^{2\gamma } \lambda \right) ,\ \\
 K_{m} J_{p} -K_{p} J_{m} &=( a_{m} -a_{p})\left( e^{-2\gamma } \lambda ^{-1} -\left( 2p^{2}\sinh( 2\gamma ) +e^{-2\gamma }\right) \lambda \right) ,\\
 K_{m}^{\prime } J_{p} -K_{p}^{\prime } J_{m} &=( a_{m} -a_{p})\left( e^{2\gamma } \lambda +\left( 2p^{2}\sinh( 2\gamma ) -e^{2\gamma }\right) \lambda ^{-1}\right) .
\end{align*}
If $a_{p} =a_{m}$ holds, then $K_p = K_m$ and $J_p = J_m$ give (iii) in Lemma \ref{lem:cond} becomes $s_ms_p=-1,$ and it does not hold because of $s_m=s_p.$ Hence, $\lambda\not\in\sigmap(U)$ for $a_{p} = a_{m}$ case.
Thus, we assume $a_{p} \neq a_{m},$ then (i) is calculated as follows:
\begin{align*}
\left( \lambda - \lambda ^{-1}\right)^{2} -4p^{2}\sinh^{2}( 2\gamma )=0.
\end{align*}
We obtain $\lambda(s_1, s_2)$ for any $(s_1, s_2)\in\{-1, +1\}^2$ is a solution to the equation.
\end{proof}

We are now in a position to state the main theorem of the present article;
\begin{theorem}
\label{main_theorem}
Let $\mathbb{R}_+, \mathbb{R}_-$ denote the set of positive and negative real numbers respectively, and let $\lambda_\pm:\mathbb{R}\rightarrow\mathbb{R}_\pm$ be two increasing bijections defined by $\lambda_\pm(x)=x\pm\sqrt{1+x^2}$ for each $x\in\mathbb{R}.$ Let $\sigma_\pm(U)$ be two subsets of $\mathbb{R}_\pm$ defined by following formula;
\begin{equation}
\label{equation: main_theorem}
\sigma_{\pm}(U)= \begin{cases}\left\{\lambda_{\pm}(+p \sinh (2 \gamma))\right\}, & a_m<\pm p_{\gamma}^{\prime}<a_p, \\ \left\{\lambda_{\pm}(-p \sinh (2 \gamma))\right\}, & a_p<\pm p_{\gamma}^{\prime}<a_m, \\ \emptyset, & \text { otherwise,}\end{cases}
\end{equation}
where the constant $p_{\gamma }^{\prime} \in (-1,1)$ is defined by
\[
p_{\gamma }^{\prime } := \frac{p}{\sqrt{p^{2} +|q|^{2}\cosh^{-2}( 2\gamma )}}.
\]
Then $\sigmap(U) = \sigma_-(U) \cup \sigma_+(U).$
\end{theorem}
\begin{proof}
From Lemma \ref{lem:condition_ev} and Lemma \ref{lem:lambda}, it is sufficient to clarify the conditions to satisfy (ii) and (iii) in Lemma \ref{lem:cond} with $\lambda = \lambda(s_1, s_2).$
We set a new sign function as follows:
\[
s_{*}^{\prime } =\sgn\left( p\cosh( 2\gamma ) -s_{2} a_{*}\sqrt{1+p^{2}\sinh^{2}( 2\gamma )}\right).
\]
Here, we do not have to consider $s_{m}^{\prime } s_{p}^{\prime } \neq 0$ case because $s_\ast'  = 0$ gives $\zeta_\ast^> = \zeta_\ast^<$ and $\lambda(s_1,s_2)\not\in\Lambda.$
Since $\lambda(s_1,s_2)\in\mathbb{R},$ we can explicitly determine $s_*$ as follows:
\begin{align*}
s_* & =\sgn\left( \lambda +\lambda ^{-1} -2p a_{*}\cosh( 2\gamma)\right)\\
 & =\sgn\left( s_{2}\sqrt{p^{2}\cosh^{2}( 2\gamma ) +| q| ^{2}} -pa_{*}\cosh( 2\gamma )\right)
 \\
 &= s_2.
\end{align*}
Therefore, we get
\begin{align*}
 & 2K_{m} K_{p}^{\prime } +2K_{p} K_{m}^{\prime } -J_{p} J_{m}\\
 & =-4\left( p\cosh( 2\gamma ) -s_{2} a_{p}\sqrt{1+p^{2}\sinh^{2}( 2\gamma )}\right)\left( p\cosh( 2\gamma ) -s_{2} a_{m}\sqrt{1+p^{2}\sinh^{2}( 2\gamma )}\right)\\
 & =-s_{m}^{\prime } s_{p}^{\prime }\sqrt{J_{p}^{2} -4K_{p} K_{p}^{\prime }}\sqrt{J_{m}^{2} -4K_{m} K_{m}^{\prime }}.
\end{align*}
It gives (iii) becomes $ s_{m}^{\prime } s_{p}^{\prime } =-1.$
Here, $s_{*}^{\prime } =+1$ means
\begin{align}
 p\cosh( 2\gamma )  >s_{2} a_{*}\sqrt{p^{2}\cosh^{2}( 2\gamma ) +| q| ^{2}}.
\end{align}
If $ s_{2} a_{*} \geq 0,$ we see that $s_{*}^{\prime}=+1$ holds if and only if $p>0$ and the following holds:
\begin{align}
\label{eq:sprime1}
(| b_{*}| p\cosh( 2\gamma ) -a_{*}| q| )(| b_{*}| p\cosh( 2\gamma ) +a_{*}| q| )  > 0.
\end{align}
On the other hand, if $ s_{2} a_{*} < 0,$ then we see that $s_\ast '=+1$ holds if and only $p \geq 0$ holds or the following holds with $p<0:$
\begin{align}
\label{eq:sprime2}
  (| b_{*}| p\cosh( 2\gamma ) -a_{*}| q| )(| b_{*}| p\cosh( 2\gamma ) +a_{*}| q| ) < 0.
\end{align}
Thus, \eqref{eq:sprime1} and \eqref{eq:sprime2} are equivalent to $|b_\ast|p\cosh(2\gamma)>|a_\ast q|$ and $|b_\ast|p\cosh(2\gamma)>-|a_\ast q|,$ respectively.
By the same way, we know similar condition for $s_\ast'=-1$ version, and these are summarized as follows.
\begin{align*}
s_\ast' = +1\ \iff \ \frac{p\cosh( 2\gamma )}{|q|} > \frac{s_{2} a_{*}}{|b_{*}|},
\qquad
s_\ast' = -1\ \iff \ \frac{p\cosh( 2\gamma )}{|q|} < \frac{s_{2} a_{*}}{|b_{*}|},
\end{align*}
where $\iff$ denotes ``if and only if''.
Additionally, we define monotonically increasing and bijective function $ f$ as
\[
f( x) =\frac{x}{\sqrt{1+x^{2}}} ,\quad x\in \mathbb{R}.
\]
Since $ f\left(\frac{s_2a_{*}}{| b_{*}| }\right) =s_2a_{*}$ and $f\lr{\frac{p}{|q|}\cosh(2\gamma)}=p_\gamma',$ applying this function to above inequalities gives the condition of (iii) in Lemma \ref{lem:cond}, i.e., $s_p' s_m' = -1,$ is equivalent to the following:
\begin{itemize}
    \item $s_p' = +1$ case:\quad $s_2 a_p < p_\gamma' < s_2 a_m,$
    \item $s_p' = -1$ case:\quad $s_2 a_m < p_\gamma' < s_2 a_p.$
\end{itemize}
Note that these inequalities guarantees $\lambda(s_1, s_2)\in \Lambda$ is certain to hold by Lemma \ref{lem:lambda}.

Next, we focus the condition of (ii).
We have
\begin{align*}
 & K_{m} J_{p} -K_{p} J_{m}\\
 & =2( a_{p} -a_{m}) p\sinh( 2\gamma )\left( s_{1}\left( e^{-2\gamma } +p^{2}\sinh( 2\gamma )\right) +s_{2} p\sqrt{1+p^{2}\sinh^{2}( 2\gamma )}\right)
\end{align*}
and
\begin{align*}
 & s_{p} K_{m}\sqrt{J_{p}^{2} -4K_{p} K_{p}^{\prime }} +s_{m} K_{p}\sqrt{J_{m}^{2} -4K_{m} K_{m}^{\prime }}\\
 & =-2s_{1} s_{2} s_{p}^{\prime }( a_{p} -a_{m}) p\sinh( 2\gamma )\left( s_{1}\left( e^{-2\gamma } +p^{2}\sinh( 2\gamma )\right) +s_{2} p\sqrt{1+p^{2}\sinh^{2}( 2\gamma )}\right).
\end{align*}
Thus, (ii) holds only if $p\sinh(2\gamma)=0$ or $s_{p}^{\prime } =-s_{1} s_{2}.$
At first, we consider $s_{p}^{\prime } =-s_{1} s_{2}$ case.
Then, above mentioned condition of (iii) derives the following list:
\begin{align*}
\lambda(+1,+1) \in \sigmap ( U) & \iff 
\phantom{-}
a_m < p_\gamma ' < a_p,\\
\lambda(+1,-1) \in \sigmap ( U) & \iff 
-a_p < p_\gamma ' < -a_m,\\
\lambda(-1,+1) \in \sigmap ( U) & \iff 
\phantom{-}
a_p < p_\gamma ' < a_m,\\
\lambda(-1,-1) \in \sigmap ( U) & \iff 
-a_m < p_\gamma ' < -a_p.
\end{align*}
This list derives desired conclusion of the theorem.
The rest of the proof is $p\sinh(2\gamma)=0$ case.
In this case, we see $\lambda(s_1, s_2) = s_2,$ so the independence of $s_1$ gives the same conclusion of the other case.
Thus, the proof is completed.
\end{proof}

\section{Discussion}
\label{section: discussion}

\subsection{Symmetry protection of eigenstates} 
\label{section: bulk-edge}
What follows is a brief summary of the existing index theory for chirally symmetric unitary operators (see, for example, \cite[\textsection 2]{Matsuzawa-Seki-Tanaka-2021}). If $U$ is an abstract \textit{unitary} operator on a Hilbert space $\cH$ satisfying \eqref{equation: chiral symmetry}, then it admits the following block-operator matrix representation with respect to $\cH = \ker(\varGamma - 1) \oplus \ker(\varGamma + 1):$
\begin{equation}
\label{equation: standard representation of U} 
U = 
\begin{pmatrix}
R_1 & iQ_2 \\
iQ_1 & R_2
\end{pmatrix}_{\ker(\varGamma - 1) \oplus \ker(\varGamma + 1)},
\end{equation}
where $R_j$ is self-adjoint for each $j = 1,2,$ and where $Q_1^* = Q_2.$ With \eqref{equation: standard representation of U} in mind, we introduce the following three formal indices:
\begin{align}
\label{equation: definition of GW indices}
\ind_\pm(\varGamma, U) &:= \dim \ker (R_1 \mp 1) - \dim \ker (R_2 \mp 1), \\
\label{equation: definition of Witten index}
\ind(\varGamma, U) &:= \dim \ker Q_1 - \dim \ker Q_2,
\end{align}
where the index defined by \eqref{equation: definition of Witten index} is the one we have previously discussed in \textsection \ref{section: introduction}. A key step in making these formal indices precise lies in the following equality:
\begin{align}
\label{equation: ker U and Rj}
\ker(U \mp 1) &= \ker(R_1 \mp 1) \oplus \ker(R_2 \mp 1). 
\end{align}
Indeed, it follows from \eqref{equation: ker U and Rj} that if the eigenspace $\ker(U \mp 1)$ is finite-dimensional, then $\ind_\pm(\varGamma, U)$ defined by \eqref{equation: definition of GW indices} is a well-defined integer, and we obtain the following estimate;
\begin{equation}
\label{equation2: topological protection of bounded states}
|\ind_\pm(\varGamma,U)| \leq \dim \ker(U \mp 1).
\end{equation}
In particular, the condition $\pm 1 \notin \ess(U)$ ensures that $\ind_\pm(\varGamma, U)$ is well-defined. Note also that \eqref{equation2: topological protection of bounded states} is a refinement of \eqref{equation1: topological protection of bounded states}, since $\ind(\varGamma, U)$ defined by \eqref{equation: definition of Witten index} is given by the following equality (see \cite[Lemma 2.1 (ii)]{Matsuzawa-Seki-Tanaka-2021} for details):
\[
\ind(\varGamma, U) = \ind_-(\varGamma, U) + \ind_+(\varGamma, U),
\]
provided that $\ker(U - 1)$ and $\ker(U + 1)$ are \textit{both} finite-dimensional.

The estimate \eqref{equation2: topological protection of bounded states} can be understood as an abstract form of the \textit{symmetry protection of eigenstates} \cite{Cedzich-Grunbaum-Stahl-Velazquez-Werner-Werner-2016,Cedzich-Geib-Grunbaum-Stahl-Velazquez-Werner-Werner-2018,Cedzich-Geib-Stahl-Velazquez-Werner-Werner-2018}. The present article is motivated by the desire to obtain a variant of \eqref{equation2: topological protection of bounded states} for \textit{non-unitary} $U.$ The hindrance of this non-unitary setting is two-fold. Firstly, the two indices $\ind_\pm(\varGamma, U)$ on the left hand side of \eqref{equation2: topological protection of bounded states} become ill-defined, if $U$ is non-unitary. That is, we need to extend the formula \eqref{equation: definition of GW indices} to the non-unitary setting in a mathematically rigorous fashion, but this is an open problem to the best of authors' knowledge. Secondly, the eigenspace $\ker(U \mp 1)$ on the right hand side of \eqref{equation2: topological protection of bounded states} also requires a certain non-trivial modification. The main theorem of the present article, Theorem \ref{main_theorem}, gives us some insight as to how this generalisation can be done under the setting of the two-phase non-unitary split-step quantum walk on the one-dimensional integer lattice.

\subsection{The spectral mapping theorem for chirally symmetric non-unitary operators} 
\label{section: smt}


Let us briefly recall the \textit{spectral mapping theorem for chirally symmetric unitary operators} \cite{Higuchi-Konno-Sato-Segawa-2014,Segawa-Suzuki-2016,Segawa-Suzuki-2019}. Let $\cH$ be an underlying separable Hilbert space, and let $U$ be an abstract unitary operator on $\cH$ satisfying the chiral symmetry condition \eqref{equation: chiral symmetry}. We can then  canonically decompose the unitary self-adjoint operator $\varGamma' := \varGamma U$ as the difference $\varGamma' = \partial^*\partial - (1 - \partial^*\partial)$ for some operator $\partial$ from $\cH = \ker(\varGamma' - 1) \oplus \ker(\varGamma' + 1)$ into an auxiliary Hilbert space $\cK,$ satisfying $\partial \partial^* = 1$ (see, for example, \cite[Lemma 3.3]{Suzuki-2019}). Here, $\partial^*\partial$ (resp. $1 - \partial^*\partial$) is the orthogonal projection onto $\ker(\varGamma' - 1)$ (resp. $\ker(\varGamma' + 1)$). The spectral mapping theorem states that for each eigenvalue $z$ of $U,$ we have
\begin{equation}
\label{equation: spectral mapping theorem for eigenvalues}
\dim \ker(U - z)
=
\begin{cases}
\dim \ker(\partial \varGamma \partial^* \mp 1) + \dim \left(\ker(\varGamma \pm 1) \cap \ker \partial\right), & z = \pm 1, \\
\dim \ker\left(\partial \varGamma \partial^* - \frac{z + z^*}{2} \right), & \mbox{otherwise,}
\end{cases}
\end{equation}
where the spectrum of the self-adjoint operator $T := \partial \varGamma \partial^*$ is a subset of $[-1,1],$ since $\|T\| \leq 1$ immediately follows from $\|\varGamma\| = 1$ and from $\|\partial\|^2 = \|\partial^*\|^2 = \|\partial^* \partial\| = 1.$ If $z \neq \pm 1,$ then the formula \eqref{equation: spectral mapping theorem for eigenvalues} has a visual interpretation (see Figure \ref{figure: smt}).
\begin{figure}
\centering
\begin{tikzpicture}
\begin{axis}[ticks=none, xmin=-1.5, xmax=1.5, ymin= -1.5, ymax=1.5, legend pos = north west, axis lines=center, xlabel=$\Re$, ylabel=$\Im$, xlabel style={anchor = north}
, width = 0.4\textwidth, height = 0.4\textwidth, clip=false
]
	\addplot [domain=0:2*pi,samples=50, smooth]({cos(deg(x))},{sin(deg(x))}); 
	\node[circle,fill, inner sep= 1.5pt, label=above:$z$] at ({1/2}, {sqrt(3)/2}) { };
	\node[circle,fill, inner sep= 1.5pt, label=below:$z^*$] at ({1/2}, {-sqrt(3)/2}) { };
	\node[circle, fill, inner sep= 1.5pt] at ({1/2}, 0) { };
 	\node at ({1/4}, {1/4}) {$\frac{z + z^*}{2}$};
    
    \draw[-stealth,densely dashed] ({1/2}, {sqrt(3)/2}) -- ({1/2}, {0.05}); %
    \draw[-stealth,densely dashed] ({1/2}, {-sqrt(3)/2}) -- ({1/2}, {-0.05}); %
\end{axis}
\end{tikzpicture}
\caption{
Given the chirally symmetric unitary operator $U  = \varGamma \varGamma',$ where $\varGamma' = \partial^*\partial - (1 - \partial^*\partial),$ we have that $z \in \T$ is an eigenvalue of $U$ if and only if so is $z^*.$  In this case, their real part $(z + z^*)/2$ turns out to be an eigenvalue of the self-adjoint operator $T = \partial \varGamma \partial^*,$ provided that $z \neq \pm 1.$
}
\label{figure: smt}
\end{figure}
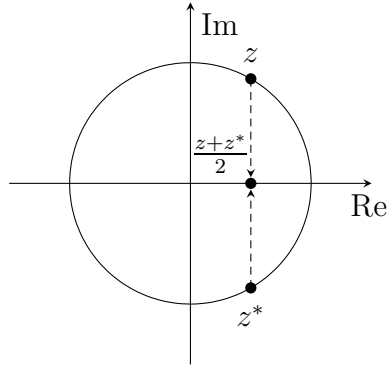

It is not known whether or not \eqref{equation: spectral mapping theorem for eigenvalues} can be naturally extended to the setting of non-unitary $U,$ since the spectrum of $U$ can be any subset of the complex plane $\C.$ This is a gap in the existing mathematics literature. With the notation introduced in Theorem \ref{main_theorem}, we can view $\gamma$ as a continuous variable. If $\gamma = 0,$ then $U$ is unitary, and we can choose and fix $p, a_m, a_p \in (-1,1)$ to ensure that either $-1$ or $+1$ becomes an eigenvalue of the unitary operator $U.$ As we continuously alter the value of $\gamma$ from $\gamma = 0,$ the operator $U$ becomes non-unitary, but we can still keep track of the continuous movement of this eigenvalue on the real axis according to the explicit formula \eqref{equation: main_theorem}. 


\newcommand{\etalchar}[1]{$^{#1}$}

\end{document}